\documentclass[12pt]{article}
\usepackage{latexsym}
\usepackage{amsfonts}
\usepackage{amssymb}
\usepackage{amsmath}

\newcommand{\R}{{\mathbb{R}}} 
\newcommand{\C}{{\mathbb{C}}} 
\newcommand{\hs}{{\cal H}}

\newcommand{\ts}{\textstyle}

\newcommand{\ket}[1]{|#1\rangle} 
\newcommand{\bra}[1]{\langle#1|} 

\newcommand{\fto}[1]{{ a_{#1}}} 
\newcommand{\ftoa}[1]{{ a^\dagger_{#1}}} 
\newcommand{\ftog}[2]{{ b_{#1}^{#2}}} 
\newcommand{\ftoag}[2]{{ b^{#2\dagger}_{#1}}} 
\newcommand{\ftogo}[2]{{ \tilde b_{#1}^{#2}}} 
\newcommand{\ftoago}[2]{{ \tilde b^{#2\dagger}_{#1}}} 

\newcommand{\gvac}[1]{{\ket{0}_{\!#1}}} 

\newcommand{\doubleGreen} {{B}} 
\newcommand{\ftAg}[3]{{ \doubleGreen^{(#2)}_{#1#3}}} 
\newcommand{\ftAag}[3]{{ \doubleGreen^{(#2)\dagger}_{#1#3}}} 
\newcommand{\epm}[2]{{ e^{(#1)}_{#2}}} 

\newcommand{\monom}[1]{{ {\cal P}(#1) }}

\newcommand{\ce} {{\epsilon}} 
\newcommand{\osplw} {{\underline \lambda}} %
\newcommand{\ospsign} {{\Lambda}} %
\newcommand{\ospsigs} {{\Lambda}} 
\newcommand{\ospvec} {{l}} %
\newcommand{\sohw} {{\overline \mu}} %
\newcommand{\sow} {{\mu}} %
\newcommand{\sosign} {{M}} %
\newcommand{\sovec} {{m}} %
\newcommand{\mullab}[1]{{ \eta_{#1} {} }} 
\newcommand{\mulmax}[1]{{ N_{#1} {} }} 

\newcommand{\orb} {{\bf o}} 
\newcommand{\spin} {{\bf s}} 
\newcommand{\hsorb} {{\hs^{\orb}}} %
\newcommand{\hsspin} {{\hs^{\spin}}} %
\newcommand{\CSAosp} {{\mathfrak{h}_{osp}}} 
\newcommand{\Gosp} {\mathfrak{g}_{osp}} 
\newcommand{\CSAso} {{\mathfrak{h}_{so}}} 
\newcommand{\Gso} {{\mathfrak{g}_{so}}} 
\newcommand{\phalf} {{q}} %
\newcommand{\V} {{V}} 

\newtheorem{theorem}{Theorem}[section]
\newtheorem{lemma}[theorem]{Lemma}
\newtheorem{corollary}[theorem]{Corollary}
\newenvironment{proof}[1][Proof]{\begin{trivlist} \item[\hskip \labelsep {\bfseries #1}]}{\end{trivlist}}
\newcommand{\qed}{\nobreak \ifvmode \relax \else
      \ifdim\lastskip<1.5em \hskip-\lastskip
      \hskip1.5em plus0em minus0.5em \fi \nobreak
      \vrule height0.75em width0.5em depth0.25em\fi}

\begin{document}
\title{Role of the orthogonal group in construction of $osp(1|2n)$ representations}

\author{Salom Igor, Institute of Physics, University of Belgrade}

\maketitle

\begin{abstract}
It is well known that the symmetric group has an important role (via Young tableaux formalism) both in labelling of the representations of the unitary group and in construction of the corresponding basis vectors (in the tensor product of the defining representations). We show that orthogonal group has a very similar role in the context of positive energy representations of $osp(1|2n, \R)$. In the language of parabose algebra, we essentially solve the long standing problem of reducibility of Green's ansatz representations.
\end{abstract}

\section{Introduction}

The $osp(1|2n, \R)$ superalgebra attracts nowadays significant attention, primarily as a natural generalization of the conformal supersymmetry in higher dimensions \cite{ExamplesFronsdal, ExamplesBandos2000, ExamplesBandos2001, ExamplesLukierski, ExamplesVasiliev2002, ExamplesPlyshchay, ExamplesVasiliev2008, SalomFDP, ExamplesFedoruk2013}. In the context of space-time supersymmetry, knowing and understanding unitary irreducible representations (UIR's) of this superalgebra is of extreme importance, as these should be in a direct relation with the particle content of the corresponding physical models.

And the most important from the physical viewpoint are certainly, so called, positive energy UIR's, which are the subject of this paper. More precisely, the goal of the paper is to clarify how these representations can be obtained by essentially tensoring the simplest nontrivial positive energy UIR (the one that corresponds to oscillator representation). This parallels the case of the UIR's of the unitary group $U(n)$ constructed within the tensor product of the defining (i.e.\ "one box") representations. In both cases the tensor product representation is reducible, and while this reduction in the $U(n)$ case is governed by the action of the commuting group of permutations, in the $osp$ case\footnote{We will often write shortly $osp(1|2n)$ or $osp$ for the $osp(1|2n, \R)$.}, as we will show, the role of permutations is played by an orthogonal group. We will clarify the details of this reduction.

The $osp(1|2n)$ superalgebra is also known by its direct relation to parabose algebra \cite{Green, OSPVeza}. In the terminology of parastatistics, the tensor product of oscillator UIR's is known as the Green's ansatz \cite{GreensAnsatzAsTensorProduct}. The problem of the decomposition of Green's ansatz space to parabose (i.e.\ $osp(1|2n)$) UIR's is an old one \cite{GreensAnsatzAsTensorProduct}, that we here solve by exploiting additional orthogonal symmetry of a "covariant" version of the Green's ansatz.

The paper is organized as follows:

In the following section we introduce the "covariant" version of the Green's ansatz (i.e.\ tensor product representation adapted to the supersymmetry case). In the third section we point out to the additional, to be called "gauge" symmetry of the ansatz, that is of orthogonal type. In the fourth section we introduce root systems of both $osp$ and the gauge algebras. The main theorems will be proved in the fifth section, clarifying the interplay between the gauge symmetry of the ansatz and the $osp$ superalgebra itself. It will be shown that:  1) quantum numbers of the gauge group remove all degeneracy of $osp$ representations appearing in the reduction of the tensor product representation; 2) transformation properties of a vector w.r.t.\ (with respect to) the gauge group determine also its $osp$ representation and {\it vice versa}, and the explicit relation between $osp$ and gauge UIR labels will be given; 3) decomposition of any $osp(1|2n, \R)$ UIR to its $sp(2n, \R)$ subrepresentations is also determined by the gauge transformation properties; 4) the listed properties allow us to explicitly write down $osp$ lowest weight vectors. The final section is reserved for some concluding remarks.

\section{Covariant Green's ansatz}

Structural relations of $osp(1|2n)$ superalgebra can be compactly written in the form of trilinear relations of odd algebra operators $\fto{\alpha}$ and $\ftoa{\alpha}$:
\begin{eqnarray}
& {}[\{ \fto{\alpha}, \ftoa{\beta}\}, \fto{\gamma}] =
-2\delta_{\beta\gamma} \fto{\alpha}, \qquad
{}[\{ \ftoa{\alpha}, \fto{\beta}\}, \ftoa{\gamma}] =
2\delta_{\beta\gamma} \ftoa{\alpha},
 \label{parabose algebra_12} \\ %
& [\{ \fto{\alpha}, \fto{\beta}\}, \fto{\gamma}], \qquad [\{ \ftoa{\alpha}, \ftoa{\beta}\}, \ftoa{\gamma}] = 0,
\label{parabose algebra_34}  \end{eqnarray} %
where operators $\{ \fto{\alpha}, \ftoa{\beta}\}$, $\{ \fto{\alpha}, \fto{\beta}\}$ and $\{ \ftoa{\alpha}, \ftoa{\beta}\}$ span the even part of the superalgebra and Greek indices take values $1, 2, \dots n$ (relations obtained from these by use of Jacobi identity are also implied). This compact notation emphasises the direct connection \cite{OSPVeza} of $osp(1|2n)$ superalgebra with the parabose algebra of $n$ pairs of creation/annihilation operators \cite{Green}.

If we (in the spirit of original definition of parabose algebra \cite{Green}) additionally require that the dagger symbol $\dagger$ above denotes hermitian conjugation in the algebra representation Hilbert space (of positive definite metrics), then we have effectively constrained ourselves to the, so called, positive energy UIR's of $osp(1|2n)$.\footnote{Omitting a short proof, we note that in such a Hilbert space all superalgebra relations actually follow from one single relation -- the first or the second of (\ref{parabose algebra_12}).} Namely, in such a space, "conformal energy" operator
\begin{equation} E\equiv \frac 12 \sum_\alpha \{ \fto{\alpha}, \ftoa{\alpha} \} \label{ConformalEnergy}\end{equation}
must be a positive operator. Operators $\fto{\alpha}$ reduce the eigenvalue of $E$, so the Hilbert space must contain a subspace that these operators annihilate. This subspace is called vacuum subspace:
\begin{equation} V_0 = \{\ket{v}, \fto{\alpha} \ket{v} = 0\} \label{VacuumSubspace}. \end{equation}
If the positive energy representation is irreducible, all vectors from $V_0$ have the common, minimal eigenvalue $\ce_0$ of $E$: $E\ket{v} = \ce_0 \ket{v}, \ket{v} \in V_0$. Representations with one dimensional subspace $V_0$ are called "unique vacuum" representations.

In this paper we will constrain our analysis to UIR's with integer and half-integer values of $\ce_0$ (in principle, $\ce_0$ has also continuous part of the spectrum -- above the, so called, first reduction point of the Verma module). It turns out that all representations from this class can be obtained by representing the odd superalgebra operators $\fto{}$ and $\ftoa{}$  as the following sum:
\begin{equation} \textstyle \fto{\alpha} = \sum_{a = 1}^{p} \ftog{\alpha}{a} \ e^{a}, \qquad \ftoa{\alpha} = \sum_{a = 1}^{p} \ftoag{\alpha}{a} \ e^{a}. \label{GGAnsatz} \end{equation}
In this expression integer $p$ is known as the order of the parastatistics, $e^a$ are elements of a real Clifford algebra:
\begin{equation} \{ e^a, e^b\} = 2 \delta^{ab} \label{CliffordAlgebra} \end{equation}
and operators $\ftog{\alpha}{a}$ together with adjoint $\ftoag{\alpha}{a}$ satisfy ordinary
bosonic algebra relations. There are total of $n\cdot p$ mutually commuting pairs of bosonic annihilation-creation operators $(\ftog{\alpha}{a},\ftoag{\alpha}{a})$:
\begin{equation}   [\ftog{\alpha}{a}, \ftoag{\beta}{b}]
= \delta_{\beta\alpha} \delta^{ab}; \quad [\ftog{\alpha}{a},
\ftog{\beta}{b}] = 0. \label{BoseAlgebra}\end{equation}
Indices $a, b, ...$ from the beginning of the Latin alphabet will, throughout the paper, take values $1,2,\dots p$.

This construction appeared long ago in a paper by Greenberg and Macrea \cite{GreenbergGauge}, where they discussed a gauge-invariant formulation of parastatistics. However, its potential relevance for the construction of representations remained unnoticed. Relation (\ref{GGAnsatz}) is a slight variation, more precisely, realization\footnote{A possibility of such a realisation is mentioned already in \cite{Green}.}, of a more common form of the Green's ansatz \cite{Green, GreenbergCompleteness}. Namely, in the same paper \cite{Green} where Green introduced parabose algebra, he offered a class of solutions ("Green's ansatz") for the trilinear relations (\ref{parabose algebra_12},\ref{parabose algebra_34}) in the terms of sum of operators satisfying "mixed" commutation and anticommutation relations:
\begin{equation} \textstyle \fto{\alpha} = \sum_{a = 1}^{p}  \ftogo{\alpha}{a}, \label{GreensAnsatzGenuine} \end{equation}
where $\ftogo{\alpha}{a}$ and $\ftoago{\alpha}{a} $ anticommute for different values of Green's indices $a$ and $b$:
\begin{equation} a \neq b \Rightarrow \{\ftoago{\alpha}{a}, \ftoago{\alpha}{b} \} = \{\ftogo{\alpha}{a}, \ftogo{\alpha}{b} \} = \{\ftogo{\alpha}{a}, \ftoago{\alpha}{b} \} = 0 \end{equation}
and behave as usual bose creation and annihilation operators otherwise:
\begin{equation} [\ftogo{\alpha}{a}, \ftoago{\beta}{a}]
= \delta_{\beta\alpha}, \; [\ftogo{\alpha}{a},
\ftogo{\beta}{a}] = 0.\end{equation}

The relation of the original Green's ansatz (\ref{GreensAnsatzGenuine}) with the Greenberg-Macrea variant (\ref{GGAnsatz}) is direct, with the latter being one concrete realization of the former: $\ftogo{\alpha}{a} = \ftog{\alpha}{a} \ e^{a}$ (we do not use any summation convention). We will refer to the (\ref{GGAnsatz}) as the \emph{covariant Green's ansatz} (it is covariant w.r.t.\ gauge group introduced in next section), or, simply, Green's ansatz, while the form (\ref{GreensAnsatzGenuine}) we will call "original" or "noncovariant".

The representation space of operators (\ref{GGAnsatz}) can be seen as tensor product of $p$ multiples of Hilbert spaces $\hs_{a}$ of ordinary linear harmonic oscillator in $n$-dimensions multiplied by the representation space of the Clifford algebra:
\begin{equation} \hs = \hs_{1}
\otimes \hs_{2} \otimes \cdots \otimes \hs_{p} \otimes \hs_{CL}. \label{HilbertSpace}\end{equation}
A single factor Hilbert space $\hs_{a}$ is the space of unitary representation of $n$ dimensional bose algebra of operators $(\ftog{\alpha}{a},\ftoag{\alpha}{a}), \alpha=1,2,\dots n$: $\hs_{a}\cong {\cal U}(\ftoag{}{a}) \gvac{a}$, where $\gvac{a}$ is the usual Fock vacuum of factor space $\hs_{a}$. The representation space $\hs_{CL}$ of real Clifford algebra (\ref{CliffordAlgebra}) is of dimension $2^{[p/2]}$, i.e.\ isomorphic with $\C^{2^{[p/2]}}$ (matrix representation). Positive definite scalar product is introduced in usual way in each of the factor spaces, endowing entire space $\hs$ also with positive definite scalar product. The space is spanned by the vectors:
\begin{equation} \hs = l.s.\{\monom{\ftoag{}{} }\gvac{} \otimes \omega \}, \label{HSBasis}\end{equation}
where $\monom{\ftoag{}{} }$ are monomials in mutually commutative operators $\ftoag{\alpha}{a}$, $\gvac{} \equiv \gvac{1} \otimes \gvac{2} \otimes \cdots \otimes \gvac{p}$ and $w \in \hs_{CL}$.

In the case $p=1$ (the Clifford part becomes trivial) we obtain the simplest positive energy UIR of $osp(1|2n)$ -- the $n$ dimensional harmonic oscillator representation. The order $p$ Green's ansatz representation of $osp(1|2n)$ is, effectively, representation in the $p$-fold tensor product of oscillator representations \cite{GreensAnsatzAsTensorProduct}, with the Clifford factor space taking care of the anticommutativity properties of odd superalgebra operators. It is easily verified that even superalgebra elements act trivially in the Clifford factor space and that their action is simply sum of actions in each of the factor spaces:
\begin{eqnarray} & \{ \ftoa{\alpha}, \fto{\beta} \} = \sum_{a=1}^p \{\ftoag{\alpha}{a}, \ftog{\beta}{a}\}, & \nonumber \\ & \{ \fto{\alpha}, \fto{\beta} \} = \sum_{a=1}^p \{\ftog{\alpha}{a},\ftog{\beta}{a}\}, \qquad \{ \ftoa{\alpha}, \ftoa{\beta} \} =  \sum_{a=1}^p \{\ftoag{\alpha}{a}, \ftoag{\beta}{a}\}.& \label{EvenOperators}\end{eqnarray}

The space (\ref{HilbertSpace}) is highly reducible under action of $osp$ superalgebra, and clarifying the details of this decomposition is one of the goals of this paper. It necessarily decomposes into direct sum of positive energy representations (both unique vacuum and non unique vacuum representations) and thus, from the aspect of $osp$ transformation properties, space $\hs$ is spanned by:
\begin{equation} \hs = l.s.\{ \ket{(\ospsign, \ospvec), \mullab{\ospsign} } \}, \label{OSPcontent} \end{equation}
where $\ospsign$ labels $osp(1|2n)$ positive energy UIR,  $\ospvec$ uniquely labels a concrete vector within the UIR $\ospsign$, and $\mullab{\ospsign} = 1, 2, \dots \mulmax{\ospsign}$ labels possible multiplicity of UIR $\ospsign$ in the representation space $\hs$. If some UIR $\ospsign$ does not appear in decomposition of $\hs$, then the corresponding $\mulmax{\ospsign}$ is zero. Label $\ospsign$ in (\ref{OSPcontent}) runs through all (integer and halfinteger positive energy) UIR's of $osp(1|2n)$ such that $\mulmax{\ospsign}>0$ and $\ospvec$ runs through all vectors from UIR $\ospsign$.

\section{Gauge symmetry of the ansatz}

Green's ansatz in the form (\ref{GGAnsatz}) possesses certain intrinsic symmetries. First, we note that hermitian operators
\begin{equation} G^{ab} \equiv \sum_{\alpha = 1}^n i (\ftoag{\alpha}{a}\ftog{\alpha}{b} - \ftoag{\alpha}{b}\ftog{\alpha}{a}) + \frac{i }4 [e^a, e^b] \label{SOpSymmetry}\end{equation}
commute with entire $osp$ superalgebra, which immediately follows after checking that $[G^{ab}, \fto{\alpha}] = 0$. Operators $G^{ab}$ themselves satisfy commutation relations of $so(p)$ algebra. The second term in (\ref{SOpSymmetry}) acts in the the Clifford factor space, generating a faithful representation of $Spin(p)$ (i.e.\ spinorial representation of double cover of $SO(p)$ group). Action of the first terms from (\ref{SOpSymmetry}) generate $SO(p)$ group action in the space $\hs_{1}\otimes \hs_{2} \otimes \cdots \otimes \hs_{p}$. In the entire space $\hs$ operators $G$ generate $Spin(p)$ group and all vectors belong to spinorial unitary representations of this symmetry group. The two terms in (\ref{SOpSymmetry}) thus resemble orbital and spin parts of rotation generators and we will often use that terminology. In particular $\hs \equiv \hsorb \otimes \hsspin$, where $\hsorb = \hs_{1}\otimes \hs_{2} \otimes \cdots \otimes \hs_{p}$ and $\hsspin = \hs_{CL}$.

For even values of $p$, in addition to the symmetry generated by operators $G$, Green's ansatz is also invariant to inversions induced by:
\begin{equation} I^a \equiv I^a_\orb \otimes I^a_\spin; \quad I^a_\orb \equiv exp(i\pi \ts \sum_\alpha \ftoag{\alpha}{a} \ftog{\alpha}{a}); \quad I^a_\spin \equiv -i \overline e e^a, \label{InversionOperators}\end{equation}
where $\overline e \equiv i^{[p/2]} e^1 e^2 \cdots e^p$. Since $(I^a)^2 = (I^a_\spin)^2 = (I^a_\orb)^2 = 1$ and $I^a$ "inverts" both $\ftog{\alpha}{a}$ and $e^a$:
\begin{equation} I^a \ftog{\alpha}{b} I^a = (-1)^{\delta_{ab}}\ftog{\alpha}{b}, \qquad I^a e^b I^a = (-1)^{\delta_{ab}}e^b, \end{equation}
it is easily verified that $[I^a, \fto{\alpha}] = 0$.

Operators $I^a$, together with $Spin(p)$ group elements generated by $G^{ab}$, form $Pin(p)$ group (the double cover of orthogonal group $O(p)$).

For odd values of $p$ linear operator with properties of $I^a_\spin$ does not exist in the universal enveloping algebra of the Clifford algebra.

We will refer to the symmetry group of the Green's ansatz, i.e.\ $Spin(p)$ for $p$ odd $Pin(p)$ for $p$ even -- as the gauge group. One reason is that it was also terminology used in the paper \cite{GreenbergGauge} where the construction (\ref{GGAnsatz}) was first explicitly introduced. The other reason is that this symmetry introduces a type of "non-physical" degree of freedom, being a symmetry of the used mathematical tool rather than $osp(1|2n)$ superalgebra itself.

Vectors in space $\hs$ carry quantum numbers also according to their transformation properties under the gauge group. As the gauge group commutes with $osp(1|2n)$, these numbers certainly remove at least a part of degeneracy of $osp$ representations in $\hs$, in the sense that relation (\ref{OSPcontent}) can be rewritten as:
\begin{equation} \hs = l.s.\{ \ket{(\ospsign, \ospvec), (\sosign, \sovec), \mullab{(\ospsign, \sosign)} } \}, \label{OSPandSOcontent} \end{equation}
where ($\ospsign$, $\ospvec$) uniquely label vector $\ospvec$ within $osp(1|2n)$ positive energy UIR $\ospsign$, ($\sosign$, $\sovec$) uniquely label vector $\sovec$ within finite dimensional UIR $\sosign$ of the gauge group, and $\mullab{(\ospsign, \sosign)} = 1, 2, ... \mulmax{(\ospsign, \sosign)}$ labels possible remaining multiplicity of tensor product of these two representations ${\cal D}_\ospsign^{osp} \otimes {\cal D}_\sosign^{gauge}$ in the space $\hs$. Again, if some combination $(\ospsign, \sosign)$ does not appear in decomposition of $\hs$, then the corresponding $\mulmax{(\ospsign, \sosign)}$ is zero.

The first important result of this paper is that the gauge symmetry actually removes all degeneracy in decomposition of $\hs$ to $osp(1|2n)$ UIR's, i.e.\ that the multiplicity of $osp(1|2n)$ UIR's is fully taken into account by labeling transformation properties of the vector w.r.t.\ the gauge symmetry group. Furthermore, we will show that there is one-to-one correspondence between UIR's of $osp(1|2n)$ and of the gauge group that appear in the decomposition, meaning that transformation properties under the gauge group action automatically fix the $osp(1|2n)$ representation. We formulate this more precisely in the following theorem.

\begin{theorem} \label{Th:main theorem}
The following statements hold for the basis (\ref{OSPandSOcontent}) of the Hilbert space $\hs$:
\begin{enumerate}\itemsep -5pt
\item All multiplicities $\mulmax{(\ospsign, \sosign)}$ are either 1 or 0.
\item Let the $\cal N$ be the set of all pairs $(\ospsign, \sosign)$ for which $\mulmax{(\ospsign, \sosign)} =1$, i.e.\ ${\cal N} = \{(\ospsign, \sosign) | \mulmax{(\ospsign, \sosign)} =1 \}$ and let the $\cal L$ and $\cal M$ be sets of all $\ospsign$ and $\sosign$, respectively, that appear in any of the pairs from $\cal N$. Then pairs from $\cal N$ naturally define bijection from $\cal L$ to $\cal M$, $\cal N\!\!:\! L \rightarrow M$.

\end{enumerate}

\end{theorem}

The theorem will be proved by explicit construction of the bijection $\cal N$, after some preliminary definitions and lemmas.

\begin{corollary}

If $osp(1|2n)$ representation $\ospsign$ appears in the decomposition of the space $\hs$, then its multiplicity in the decomposition is given by the dimension of the gauge group representation ${\cal N}(\ospsign)$.

\end{corollary}

\section{Root systems}

At this point we must introduce root systems, both for $osp(1|2n)$ superalgebra and for the $so(p)$ algebra of the gauge group.

We choose basis of a Cartan subalgebra $\CSAosp$ of (complexified) $osp(1|2n)$ as:
\begin{equation} \CSAosp =  l.s.\Big\{ \ts \frac 12 \{\ftoa{\alpha}, \fto{\alpha} \}, \alpha =1, 2, \dots n \Big\}. \end{equation}
Positive roots, expressed using elementary functionals, are:
\begin{eqnarray}  \Delta^+_{osp} &=& \{+\delta_\alpha, 1 \leq \alpha \leq n; +\delta_\alpha + \delta_\beta, 1 \leq \alpha < \beta \leq n; \nonumber \\ & & +\delta_\alpha - \delta_\beta, 1 \leq \alpha < \beta \leq n; +2\delta_\alpha, 1 \leq \alpha \leq n \}  \end{eqnarray}
and the corresponding positive root vectors, spanning subalgebra $\Gosp^+$, are (in the same order):
\begin{eqnarray}
& \Big\{ \ftoa{\alpha}, 1 \leq \alpha \leq n; \{\ftoa{\alpha}, \ftoa{\beta} \}, 1 \leq \alpha < \beta \leq n; \nonumber \\
&  \{\ftoa{\alpha}, \fto{\beta} \}, 1 \leq \alpha < \beta \leq n; \{\ftoa{\alpha}, \ftoa{\alpha} \}, 1 \leq \alpha \leq n \Big\}.
\end{eqnarray}
Simple root vectors are: 
\begin{equation} \Big\{ \{\ftoa{1}, \fto{2} \}, \{\ftoa{2}, \fto{3} \}, \dots, \{\ftoa{n-1}, \fto{n} \}, \ftoa{n} \Big\}. \end{equation}
With this choice of positive roots, positive energy UIR's of $osp(1|2n)$ become lowest weight representations. Thus, we will label positive energy UIR's of $osp(1|2n)$ either by their lowest weight
\begin{equation} \osplw = (\osplw_1, \osplw_2, \dots, \osplw_n), \end{equation}
or by its signature
\begin{equation} \ospsign = [d; \ospsigs_1, \ospsigs_2, \dots, \ospsigs_{n-1}] \label{ospSignature}\end{equation}
related to the lowest weight $\osplw$ by $d = \osplw_{1}$, $\ospsigs_\alpha = \osplw_{\alpha + 1} - \osplw_{\alpha}$. $\ospsigs_\alpha$ are positive integers \cite{DobrevZhang} and spectrum of $d$ is positive and dependant of $\ospsigs_\alpha$ values. 

As a basis of Cartan subalgebra $\CSAso$ of $so(p)$ we take:
\begin{equation} \CSAso =  l.s.\bigg\{ G^{(k)}\equiv G^{2k-1, 2k}, k =1, 2, \dots \phalf \bigg\}, \end{equation}
where $\phalf = [p/2]$ is the dimension of Cartan subalgebra (indices $k,l,...$ from the middle of alphabet will take values $1,2,..., \phalf$). Positive roots in case of even $p$ are:
\begin{eqnarray}  \Delta^+_{so} &=& \{+\delta_k + \delta_l, 1 \leq k < l \leq \phalf; +\delta_k - \delta_l, 1 \leq k < l \leq \phalf \},  \end{eqnarray}
while in the odd case we additionally have $\{+\delta_k, 1 \leq k \leq \phalf \}$.

In accordance with the choice of Cartan subalgebra $\CSAso$ it is more convenient to use the following linear combinations:
\begin{equation} \ftAag{\alpha}{k}{\pm} \equiv \ts \frac{1}{\sqrt 2}(\ftoag{\alpha}{2k-1} \pm i \ftoag{\alpha}{2k}) , \qquad \ftAg{\alpha}{k}{\pm} = \ts \frac{1}{\sqrt 2}(\ftog{\alpha}{2k-1} \mp i \ftog{\alpha}{2k}), \label{Apm} \end{equation}
instead of $\ftoag{}{}$ and $\ftog{}{}$, as $[G^{(k)}, \ftAag{\alpha}{l}{\pm}]= \pm \delta^{kl} \ftAag{\alpha}{l}{\pm}$ and $[G^{(k)}, \ftAg{\alpha}{l}{\pm}]= \mp \delta^{kl} \ftAg{\alpha}{l}{\pm}$. Similarly, we introduce $\epm{k}{\pm} \equiv \ts \frac{1}{\sqrt 2}(e^{2k-1} \pm i e^{2k})$ that satisfy:
\begin{equation} [G^{(k)}, \epm{l}{\pm}]= \pm \delta^{kl} \epm{l}{\pm}. \end{equation}
Odd superalgebra operators take form:
\begin{eqnarray} & & \ftoa{\alpha} = \Big(\sum_{k=1}^\phalf \ftAag{\alpha}{k}{+}\epm{k}{-} + \ftAag{\alpha}{k}{-}\epm{k}{+}\Big) + \epsilon \, \ftoag{\alpha}{p}e^{p}, \label{CreationOperatorsInA} \\
& & \fto{\alpha} = \Big(\sum_{k=1}^\phalf \ftAg{\alpha}{k}{+}\epm{k}{+} + \ftAg{\alpha}{k}{-}\epm{k}{-}\Big) + \epsilon \, \ftog{\alpha}{p}e^{p},\label{AnnihilationOperatorsInA} \end{eqnarray}
where $\epsilon = p \mod 2$.

Now we can express root vectors as:
\begin{eqnarray}
G_{\pm \delta_k \pm \delta_l} &=& \ts \frac i2(G^{2k-1,2l-1} - G^{2k,2l} \pm i G^{2k,2l-1} \pm i G^{2k-1,2l}) \nonumber \\
&=& \ts \sum_\alpha (\ftAag{\alpha}{k}{\pm} \ftAg{\alpha}{l}{\mp} - \ftAag{\alpha}{l}{\pm} \ftAg{\alpha}{k}{\mp}) + \frac 12 \epm{k}{\pm}\epm{l}{\pm},
\end{eqnarray}
\begin{eqnarray}
G_{\pm \delta_k \mp \delta_l} &=& \ts \frac i2(G^{2k-1,2l-1} + G^{2k,2l} \pm i G^{2k,2l-1} \mp i G^{2k-1,2l}) \nonumber \\
&=& \ts \sum_\alpha (\ftAag{\alpha}{k}{\pm} \ftAg{\alpha}{l}{\pm} - \ftAag{\alpha}{l}{\mp} \ftAg{\alpha}{k}{\mp}) + \frac 12 \epm{k}{\pm}\epm{l}{\pm},
\end{eqnarray}
and, for odd $p$, there are also:
\begin{eqnarray}  G_{\pm \delta_k} = \ts \frac i2(G^{2k-1,p} \pm i G^{2k,p}) = \ts \sum_\alpha (\ftAag{\alpha}{k}{\pm} \ftog{\alpha}{p} - \ftoag{\alpha}{p} \ftAg{\alpha}{k}{\pm}) + \frac 12 \epm{k}{\pm} e^{p}. \end{eqnarray}

The space $\hs$ decomposes to spinorial UIR's of $so(p)$ with the highest weight $\sohw = (\sohw^1, \sohw^2, \dots, \sohw^q)$ satisfying $\sohw^1 \geq \sohw^2 \geq \dots \geq \sohw^{q-1} \geq |\sohw^q| \geq \frac 12$ with all $\sohw^q$ taking half-integer values ($\sohw^q$ can take negative values when $p$ is even). However, since the gauge symmetry group in the case of even $p$ is enlarged by presence of inversion operators (\ref{InversionOperators}), any highest weight of UIR of the gauge group satisfies:
\begin{equation} \sohw^1 \geq \sohw^2 \geq \dots \geq \sohw^q \geq 0.\label{sohwCondition} \end{equation}
As the gauge group representation in $\hs$ is spinorial, all $\sohw^k$ take half-integer values greater or equal to $\frac 12$. To label UIR's of the gauge group we will also use signature
\begin{equation} \sosign = [\sosign^1, \sosign^2,\dots,\sosign^q] \label{soSignature}\end{equation}
with $\sosign^k = \sohw^{k} - \sohw^{k+1}, k < q$ and $\sosign^q = \sohw^{q} - \frac 12$. All $\sosign^k$ are positive integers.

The "spin" factor space $\hsspin$ is irreducible w.r.t.\ action of the gauge group. Gauge group representation in the space $\hsspin$ has the highest weight $\sohw_\spin = (\frac 12, \frac 12,\dots, \frac 12)$. Weight spaces of this representation are one dimensional, meaning that basis vectors can be fully specified by weights $\sow_\spin$: 
\begin{equation} \hsspin = l.s.\{ \omega_{\sow_{\spin}} \equiv \omega(\sow_{\spin}^1, \sow_{\spin}^2, \dots, \sow_{\spin}^q)| \sow_\spin^k = \pm \frac 12 \}. \end{equation}
An action of operators $\epm{k}{+}, \epm{k}{-}$ and $e^p$ in this basis is given by:
\begin{equation} \epm{k}{\pm}\omega(\sow_{\spin}^1, \sow_{\spin}^2, \dots, \sow_{\spin}^q) = \sqrt 2 \Bigg(\prod_{l=1}^{k-1} 2\sow_{\spin}^l \Bigg) \omega(\sow_{\spin}^1,  \dots, \sow_{\spin}^{k-1},\sow_{\spin}^k \pm 1,\sow_{\spin}^{k+1},\dots, \sow_{\spin}^q) \end{equation}
and, when $p$ is odd, also:
\begin{equation} e^{p}\omega(\sow_{\spin}^1, \sow_{\spin}^2, \dots, \sow_{\spin}^q) = \Bigg( \prod_{l=1}^{q} 2\sow_{\spin}^l\Bigg)  \omega(\sow_{\spin}^1, \sow_{\spin}^2, \dots, \sow_{\spin}^q). \end{equation}
In these definitions it is implied that $\omega(\sow_{\spin}^1, \sow_{\spin}^2, \dots, \sow_{\spin}^q) \equiv 0$ if any $|\sow_{\spin}^k| > \frac 12$.

Gauge group representation in "orbital" factor space $\hsorb$ decomposes to highest weight $\sohw_\orb$ UIR's such that all $\sohw_\orb^k$ are positive integers. Besides, it is not difficult to verify that, if $n < q$, then
\begin{equation} \sohw_\orb^{n+1} = \sohw_\orb^{n+2} = \dots = \sohw_\orb^{q} = 0 \label{nlessq}\end{equation}
(since maximally $n$ operators (\ref{Apm}) can be antisymmetrized).

\section{Decomposition of the Green's ansatz space}

Let $\V^\orb_{\sohw_\orb}$ be a (reducible) subspace of vectors from $\hsorb$ that transform under representation $\sohw_\orb$ of the gauge group $G^\orb$. Analogously, since $\hsspin$ is irreducible, we will write $\V^\spin_{\sohw_\spin} \equiv \hsspin$. Tensor product $\V^\orb_{\sohw_\orb} \otimes \V^\spin_{\sohw_\spin}$ decomposes under action of full gauge symmetry group $G$ to subspaces $\V_{(\sohw_\orb \sohw_\spin)\sohw}$, each transforming according to gauge group representation $\sohw$:
\begin{equation} \V^\orb_{\sohw_\orb} \otimes \V^\spin_{\sohw_\spin} = \sum_{\sohw \in {\cal M}} \V_{(\sohw_\orb \sohw_\spin)\sohw}. \end{equation}
In the above relation, $\cal M$ is set of highest weights $\sohw$ of irreducible representations ${\cal D}_\sohw$ that appear in the product ${\cal D}_{\sohw_\orb} \otimes {\cal D}_{\sohw_\spin}$. Conversely, for a given $\sohw$, let ${\cal M}_\orb = {\cal M}_\orb(\sohw)$ be the set of highest weights $\sohw_\orb$ such that ${\cal D}_\sohw$ belongs to ${\cal D}_{\sohw_\orb} \otimes {\cal D}_{\sohw_\spin}$. (Since the gauge group is $Spin(p)$ or $Pin(p)$, this is the same as ${\cal M}_\orb(\sohw) = \{ \sohw_\orb | {\cal D}_{\sohw_\orb} \subset {\cal D}_{\sohw} \otimes {\cal D}_{\sohw_\spin}$\}.) It holds:
\begin{equation} \V_{\sohw} = \sum_{\sohw_\orb \in {\cal M}_\orb} \V_{(\sohw_\orb \sohw_\spin)\sohw}, \label{VmuDecomposition}\end{equation}
where $\V_{\sohw}$ is the subspace of vectors from $\hs$ that transform under representation $\sohw$ of the gauge group $G$.

Since the gauge group commutes with $osp(1|2n)$ operators, it is clear that each of the subspaces $\V_{\sohw}$ is invariant w.r.t.\ $osp(1|2n)$ superalgebra action. Furthermore, since even superalgebra operators (\ref{EvenOperators}) also commute separately with orbital and spin parts of the gauge group, each of the subspaces $\V_{(\sohw_\orb \sohw_\spin)\sohw}$ is invariant under action of $sp(2n)$ subalgebra. Only odd superalgebra elements connect different terms of (\ref{VmuDecomposition}). We will show that decomposition (\ref{VmuDecomposition}) actually reflects the decomposition of $osp(1|2n)$ irreducible representation into $sp(2n)$ subrepresentations.

First we will show that the lowest weight vector $\ket{(\osplw, \osplw), \mullab{\osplw}}$ of any $osp(1|2n)$ irreducible representation $\osplw$ (here we used notation of (\ref{OSPcontent}) with the lowest weight $\osplw$ as UIR label instead of signature $\ospsign$, and $\ospvec$ being the lowest weight vector) must belong to some subspace $\V_{(\sohw_\orb \sohw_\spin)\sohw}$ which satisfies $\sohw = \sohw_\orb + \sohw_\spin$.

To show this we introduce operator:
\begin{equation} Q \equiv \frac 12 \sum_\alpha [\fto{\alpha}, \ftoa{\alpha}].\label{Qoperator} \end{equation}
The first important property of this operator is that it commutes with subalgebra $sp(2n)$, which is easily verified. Next, substituting (\ref{GGAnsatz}) in the definition of $Q$ yields:
\begin{eqnarray}  Q & = & \sum_{a,b} \sum_\alpha (\ftoag{\alpha}{a}\ftog{\alpha}{b} - \ftoag{\alpha}{b}\ftog{\alpha}{a})\frac{[e^a, e^b]}{4} + \sum_{a,\alpha} \frac 12 (e^a)^2 = \nonumber \\
&=& \frac {np}{2} +2 \sum_{a>b} G^{ab}_\orb G^{ab}_\spin = \frac {np}{2} + \sum_{a>b} (G^{ab})^2 - (G^{ab}_\orb)^2 - (G^{ab}_\spin)^2.\label{QasLS} \end{eqnarray}
Therefore, this operator is also gauge invariant. Furthermore, we see that it is actually (up to a constant) twice a "spin-orbit coupling" operator (analogue of $\bf L\cdot S$ operator in quantum mechanics). Another important property is given by the following lemma.

\begin{lemma}
For arbitrary vector $\ket v \in \hs$ it holds: $Q \ket v = E \ket v$ if and only if $\ket v$ belongs to $V_0$, where $E$ is conformal energy (\ref{ConformalEnergy}) and $V_0$ denotes vacuum subspace (\ref{VacuumSubspace}) \label{lemmaQ}.
\end{lemma}

\begin{proof} Since metric in $\hs$ is positive definite, from:
$$ \bra v E - Q \ket v = \sum_\alpha \bra v \ftoa{\alpha} \fto{\alpha} \ket v = \sum_\alpha ||\fto{\alpha} \ket v||^2 $$
follows $Q \ket v = E \ket v \Rightarrow \forall \alpha, \fto{\alpha} \ket v = 0 \Leftrightarrow \ket v \in V_0$. Proof in the opposite direction is trivial. \qed
\end{proof}

Now consider an $osp(1|2n)$ lowest weight vector $\ket{(\osplw, \osplw), (\sohw, \sovec), \mullab{(\osplw, \sohw)} }$ from basis (\ref{OSPandSOcontent}). We prove another lemma:
\begin{lemma} Vector $\ket{(\osplw, \osplw), (\sohw, \sovec), \mullab{(\osplw, \sohw)} }$ belongs to subspace $\V_{(\sohw_\orb \sohw_\spin)\sohw}$ with $\sohw_\orb = \sohw - \sohw_\spin$.
\end{lemma}

\begin{proof} The lowest weight vector belongs to the vacuum subspace $V_0$, and thus, as a consequence of lemma (\ref{lemmaQ}):
\begin{equation} Q \ket{(\osplw, \osplw), (\sohw, \sovec), \mullab{(\osplw, \sohw)} } = E \ket{(\osplw, \osplw), (\sohw, \sovec), \mullab{(\osplw, \sohw)} }. \end{equation}
Using the known expression for Casimir operator eigenvalue as a function of highest weight, from (\ref{QasLS}) we obtain:
\begin{eqnarray}   Q \ket{(\osplw, \osplw), (\sohw, \sovec), \mullab{(\osplw, \sohw)} } & = & \Big(\frac{np}{2} + \left<\sohw, \sohw\right> - \left<\sohw_\orb, \sohw_\orb \right> - \left<\sohw_\spin, \sohw_\spin \right> + \\
& &  2\left<\rho, \sohw -\sohw_\orb - \sohw_\spin \right>\Big)\ket{(\osplw, \osplw), (\sohw, \sovec), \mullab{(\osplw, \sohw)} } ,  \nonumber \end{eqnarray}
where $\rho$ is half sum of positive roots of $so(p)$. Let us write $\sohw =  \sohw_\orb + \sohw_\spin - \sow_\Delta$. Taking into account that $\left<\sohw_\orb, \sohw_\spin \right> = \frac 12 \sum_k \sohw_\orb^k$, left hand side becomes:
\begin{eqnarray}   Q \ket{(\osplw, \osplw), (\sohw, \sovec), \mullab{(\osplw, \sohw)} } & = &  \Big(\frac{np}{2} + \sum_k \sohw_\orb^k - \left<\sow_\Delta, 2 \sohw + 2 \rho + \sow_\Delta \right>\Big)\ket{(\osplw, \osplw), (\sohw, \sovec), \mullab{(\osplw, \sohw)} } .  \nonumber \end{eqnarray} 
Using (\ref{CreationOperatorsInA},\ref{AnnihilationOperatorsInA}) we can rewrite the right-hand side as:
\begin{eqnarray}   E \ket{(\osplw, \osplw), (\sohw, \sovec), \mullab{(\osplw, \sohw)} } & = & \Big(\frac{np}{2} + (\sum_{\alpha, k} \ftAag{\alpha}{k}{+}\ftAg{\alpha}{k}{+} + \ftAag{\alpha}{k}{-}\ftAg{\alpha}{k}{-}) +  \nonumber \\
& &  \epsilon (\sum_{\alpha, k} \ftoag{\alpha}{p}\ftog{\alpha}{p})\Big)\ket{(\osplw, \osplw), (\sohw, \sovec), \mullab{(\osplw, \sohw)} } ,  \label{righthandside}\end{eqnarray}
where $\epsilon = p \mod 2$. Operator $E$ commutes not only with the full gauge group $G$, but also with its orbital part $G^\orb$ alone, and with orbital inversion operators $I_\orb^a$ (\ref{InversionOperators}). This fact can be used to evaluate expression (\ref{righthandside}) on some term with the highest weight w.r.t\ orbital gauge group and prove:
\begin{eqnarray}   E \ket{(\osplw, \osplw), (\sohw, \sovec), \mullab{(\osplw, \sohw)} } & = & \Big(\frac{np}{2} + \sum_k \sohw_\orb^k \Big)\ket{(\osplw, \osplw), (\sohw, \sovec), \mullab{(\osplw, \sohw)} } .  \label{righthandside2}\end{eqnarray}
Equating left and right sides, we obtain:
\begin{equation} \left<\sow_\Delta, 2 \sohw + 2 \rho + \sow_\Delta \right> = 0. \end{equation}
Since $\sow_\Delta^k$ can take only values $0$ and $1$ and $(2 \sohw + 2 \rho + \sow_\Delta)^k \geq 1$ this is enough to conclude that $\sow_\Delta = 0$. \qed
\end{proof}

The following lemma is the remaining step necessary to complete the proof of theorem (\ref{Th:main theorem}).

\begin{lemma} \label{Lm: vector form} The vector $\ket{(\osplw, \osplw), (\sohw, \sohw), \mullab{(\osplw, \sohw)} } \in \hs$ that is the lowest weight vector of $osp(1|2n)$ positive energy UIR $\osplw$ and the highest weight vector of the gauge group UIR $\sohw$ exists if and only if signatures $\ospsign$ and $\sosign$ (\ref{ospSignature}, \ref{soSignature}) satisfy:
\begin{equation} \sosign_k = \ospsigs_{n-k}, \label{Nbijection} \end{equation}
where $\ospsigs_0 \equiv d - p/2$ and it is implied that $M_k = 0, k > q$ and $\ospsigs_\alpha=0, \alpha < 0$.
In that case this vector has the following explicit form (up to multiplicative constant) in the basis (\ref{HSBasis}):
\begin{eqnarray} & \ket{(\osplw, \osplw), (\sohw, \sohw), \mullab{(\osplw, \sohw)}} =\Big(\ftAag{n}{1}{+}\Big)^{\ospsigs_{n-1}}
\Big(\ftAag{n}{1}{+}\ftAag{n-1}{2}{+} - \ftAag{n}{2}{+}\ftAag{n-1}{1}{+}\Big)^{\ospsigs_{n-2}} \cdots & \nonumber \\
&  \cdot \Big(\displaystyle \sum_{k_1, k_2, ... k_n = 1}^{\min(n, \phalf)} \varepsilon_{k_1 k_2 ... k_n} \ftAag{n}{k_1}{+}\ftAag{n-1}{k_2}{+} \cdots \ftAag{1}{k_n}{+} \Big)^{\ospsigs_0} \gvac{} \otimes \omega(\ts \frac 12, \frac 12, \dots, \frac 12) \label{lwhwVectorForm}. &\end{eqnarray}
\end{lemma}
\begin{proof}
As the consequence of the previous lemma, vector $\ket{(\osplw, \osplw), (\sohw, \sohw), \mullab{(\osplw, \sohw)} }$ must belong to the subspace $\V_{(\sohw_\orb \sohw_\spin)\sohw_\orb + \sohw_\spin} $. Such highest weight vector can be written as tensor product of a gauge group highest weight vector in space $\hsorb$ and the highest weight vector $\omega(\ts \frac 12, \frac 12, \dots, \frac 12)$ of the space $\hsspin$:
\begin{equation} \ket{(\osplw, \osplw), (\sohw, \sohw), \mullab{(\osplw, \sohw)} } = {\cal P}^{hw}(\doubleGreen^\dagger_{+}, \doubleGreen^\dagger_{-}, \ftoag{}{p})\gvac{} \otimes \omega(\ts \frac 12, \frac 12, \dots, \frac 12). \end{equation}
As a consequence of $\fto{\alpha} \ket{(\osplw, \osplw), (\sohw, \sohw), \mullab{(\osplw, \sohw)} } = 0$ and the form (\ref{AnnihilationOperatorsInA}) of $\fto{\alpha}$ it turns out that the polynomial ${\cal P}^{hw}$ above can contain neither $ \doubleGreen^\dagger_{-}$ nor $\ftoag{}{p}$. Thus:
\begin{equation} \ket{(\osplw, \osplw), (\sohw, \sohw), \mullab{(\osplw, \sohw)} } = {\cal P}^{hw}(\doubleGreen^\dagger_{+})\gvac{} \otimes \omega(\ts \frac 12, \frac 12, \dots, \frac 12). \end{equation}
It is now easily verified that ${\cal P}^{hw}$ must be of the form (\ref{lwhwVectorForm}) to produce vector that is of the lowest $osp$ weight and of the highest gauge group weight. The corresponding weights and UIR signatures are then directly inferred. \qed
\end{proof}

Note that the Lemma \ref{Lm: vector form} also determines whether an $osp$ representation $\ospsign$ appears or not in the decomposition of Green's ansatz of order $p$: UIR $\ospsign$ appears in the decomposition if and only if the condition (\ref{Nbijection}) can be satisfied by allowed integer values of $\sosign_k$. However, if $q$ is not sufficiently high, the first $n - q$ of the $\ospsign$ components $\ospsign_0, \ospsign_1, \dots \ospsign_{n-q-1}$ are bound to be zero.

\begin{corollary}
All (half)integer positive energy UIR's of $osp(1|2n)$ can be constructed in space $\hs$ with $p \leq 2n+1$.
\end{corollary}
\begin{proof}
Due to relation (\ref{Nbijection}), values $\ospsigs_0, \ospsigs_1, \dots \ospsigs_{n-1}$ can be arbitrary integers when $\phalf \geq n$: choice $p = 2n$ contains integer values of $d$ UIR's while $p = 2n + 1$ contains half-integer values. That spaces $\hs$ for some $p < 2n$ also contain all UIR's with $d < n$, can be verified by checking the list of all positive energy UIR's of $osp(1|2n)$ given elsewhere \cite{ospUIRclassification}. \qed
\end{proof}

In other words, the above corollary states that no additional UIR's of $osp(1|2n)$ appear when considering $p > 2n+1$, i.e.\ it is sufficient to consider only $p \leq 2n+1$.

The proof of the {\bf Theorem \ref{Th:main theorem}} now follows from the Lemma \ref{Lm: vector form}.
\begin{proof}
Lemma \ref{Lm: vector form} gives the explicit form of the vector that is the lowest weight vector of $osp(1|2n)$ positive energy UIR $\osplw$ and the highest weight vector of the gauge group UIR $\sohw$, when such vector exists. It follows that there can be at most one such vector. Therefore, the multiplicity $\mulmax{(\osplw, \sohw)}$ can be either 1 or 0. The relation between $\osplw$ and $\sohw$ is given by (\ref{Nbijection}) and it defines bijection $\cal N$. \qed
\end{proof}

Finally, let us consider an $osp$ representation ${\cal D}_\osplw$ with the lowest weight vector $\ket{(\osplw, \osplw), (\sohw, \sovec)}$ (now we omitted the multiplicity label as we have proved it is unnecessary). The representation space ${\cal U}(\Gosp^+) \ket{(\osplw, \osplw), (\sohw, \sovec)}$ contains a number of lowest weight vectors of even subalgebra $sp(2n)$, corresponding to decomposition of the $osp(1|2n)$ representation to $sp(2n)$ subrepresentations. It also holds: ${\cal U}(\Gosp^+) \ket{(\osplw, \osplw), (\sohw, \sovec)} \subset V_\sohw$.

We state the following theorem on the decomposition to even subalgebra subrepresentations:

\begin{theorem}
Each of the subspaces $\V_{(\sohw_\orb \sohw_\spin)\sohw}$, $\sohw_\orb \in {\cal M}_\orb(\sohw)$ of the space $V_\sohw$ (\ref{VmuDecomposition}) contains at least one of the $sp(2n)$ subrepresentations in the decomposition of ${\cal U}(\Gosp^+) \ket{(\osplw, \osplw), (\sohw, \sovec)}$.
\end{theorem}

\begin{proof} The theorem is a direct consequence of the fact that that each of the subspaces $\V_{(\sohw_\orb \sohw_\spin)\sohw}$ is closed w.r.t.\ action of the even subalgebra. \qed
\end{proof}

Knowing the gauge transformation properties, which are determined by this theorem, drastically simplifies the problem of finding the lowest weight vectors of $sp(2n)$ subrepresentations.

\section{Remarks}

The established properties of the gauge symmetry action in the $osp$ representation space (\ref{HilbertSpace}) essentially mean that the orthogonal group for $osp(1|2n, \R)$ indeed plays the role that the symmetric group has in the case of the unitary group $U(n)$. Each subspace $V_{(\sosign, \sovec)}$  of vectors transforming as $(\sosign, \sovec)$ w.r.t.\ gauge group is the irreducible representation space of $osp$ UIR $\ospsign = {\cal N}^{-1} (\sosign) $. This parallels the case of tensor product of $U(n)$ defining representations, where each subspace of definite permutation symmetry properties is irreducible w.r.t.\ $U(n)$ action. Either gauge group UIR label $\sosign$ or the $osp$ UIR label $\ospsign$ can be used for labeling of both gauge and $osp$ UIR's, just as any Young diagram labels both UIR's of permutation group and of $U(n)$. As an additional feature, we have shown that gauge spin-orbit coupling properties determine decomposition of $osp(1|2n, \R)$ UIR's to UIR's of the even subalgebra $sp(2n, \R)$.

A curious consequence of the obtained results is that the tensor products of up to $2n+1$ oscillatory representations already contain all UIR's that are obtainable in this way. This is in contrast to $U(n)$ analogy, since we need arbitrary number of boxes in a Young diagram to construct arbitrary $U(n)$ representation. In the parastatistics terminology, this means that considering the order of parastatistics (defined as the number of factor spaces) $p> 2n+1$ does not introduce any new representation of parabose algebra. This also holds when considering the subclass of the unique vacuum representations, and holds already for $n=1$ (however, the vacuum state in general has the form given by (\ref{lwhwVectorForm}) ). Physically, in the context of $osp$ space-time symmetry, this means that all particles belonging to (half)integer energy representations can be seen as composed of up to $2n+1$ particles of the simplest type.

We also note that, in practical applications, a gauge fixing that removes multiplicity of $osp$ representations can be easily introduced. For example, the "highest weight gauge" condition would be imposing a constraint $X \ket{v} = 0, \forall X \in \Gso^+$. The subspace of vectors satisfying such condition no longer possesses any multiplicity of $osp$ UIR's.

We constrain ourselves here only to a short comment on the relation of this covariant form of the Green's ansatz with its "non covariant" form appearing in the Green's seminal paper \cite{Green}. The relation between covariant and non-covariant Green operators is simple and invertible: $\ftogo{\alpha}{a} = \ftog{\alpha}{a} \ e^{a}$ and $\ftog{\alpha}{a} = \ftogo{\alpha}{a} \ e^{a}$. Using this relation, any vector written in the noncovariant ansatz, i.e.\ of the form ${\cal U}(\ftoago{}{})\ket{0}$ can be easily rewritten in the covariant form, and then its $osp$ properties can be inferred considering its transformation w.r.t.\ the gauge group. This mapping is not one to one, since the Fock vacuum of the noncovariant ansatz is replaced by $\ket{0}\otimes \omega $ where $\omega$ is arbitrary vector from the Clifford space. All types of $osp$ UIR's existing in the covariant ansatz of order $p$ also exist in its non-covariant counterpart. The relation of the two forms of the Green's ansatz can be explored in more detail, but one of the points of this paper is that, in physical applications, there is no particular need to consider the non-covariant version at all: the covariant form (\ref{GGAnsatz}) has much better mathematical properties.

Finally, we express our belief that the approach exposed here can be also generalized to the case of some other (super)algebras.

\section*{Acknowledgments}
\label{ack}

This work was financed by the Serbian Ministry of Science and
Technological Development under grant number OI 171031.

\end{document}